\documentclass[envcountsect]{llncs}
\usepackage{stmaryrd}
\usepackage{amsmath}
\usepackage{amssymb}
\usepackage{newlfont}
\usepackage{graphicx}
\usepackage{caption}
\usepackage{algorithm,setspace}
\usepackage{algorithmic}

\title{Quantifier Elimination over Finite Fields Using Gr\"obner~Bases\thanks{This research was sponsored by National Science
   Foundation under contracts no.~CNS0926181,
   no.~CCF0541245, and no. CNS0931985,
   the SRC under contract
   no.~2005TJ1366, General Motors under contract no.~GMCMUCRLNV301, Air 
Force
   (Vanderbilt University) under contract
   no.~18727S3, the GSRC under contract no. 1041377 
(Princeton University), the Office of Naval Research under award 
no.~N000141010188, and DARPA under contract FA8650-10-C-7077.
}}
\author{Sicun Gao\and Andr\'e Platzer \and Edmund M. Clarke }
\institute{Carnegie Mellon University, Pittsburgh, PA, USA}

\begin{document}

\maketitle
\begin{abstract} 
We give an algebraic quantifier elimination algorithm for the first-order theory over any given finite field using Gr\"obner basis methods. The algorithm relies on the strong Nullstellensatz and properties of elimination ideals over finite fields. We analyze the theoretical complexity of the algorithm and show its application in the formal analysis of a biological controller model.
\end{abstract}

\section{Introduction}

We consider the problem of quantifier elimination of first-order logic formulas in the theory $T_q$ of arithmetic in any given finite field $F_q$. Namely, given a quantified formula $\varphi(\vec x;\vec y)$ in the language, where $\vec x$ is a vector of quantified variables and $\vec y$ a vector of free variables, we describe a procedure that outputs a quantifier-free formula $\psi(\vec y)$, such that $\varphi$ and $\psi$ are equivalent in $T_q$. 

Clearly, $T_q$ admits quantifier elimination. A naive algorithm is to enumerate the exponentially many assignments to the free variables $\vec y$, and for each assignment $\vec a\in F^{|\vec y|}$, evaluate the truth value of the closed formula $\varphi(\vec x;\vec a)$ (with a decision procedure). Then the quantifier-free formula equivalent to $\varphi(\vec x;\vec y)$ is $\bigvee_{\vec a\in A}(\vec y = \vec a)$, where $A=\{\vec a\in F^{|\vec y|}: \varphi(\vec x;\vec a) \mbox{ is true.}\}$. This naive algorithm {\em always} requires exponential time and space, and cannot be used in practice. Note that a quantifier elimination procedure is more general and complex than a decision procedure: Quantifier elimination yields an equivalent quantifier-free formula  while a decision procedure outputs a yes/no answer. For instance, fully quantified formulas over finite fields can be ``bit-blasted'' and encoded as Quantified Boolean Formulas (QBF), whose truth value can, in principle, be determined by QBF decision procedures. However, for formulas with free variables, the use of decision procedures can only serve as an intermediate step in the naive algorithm mentioned above, and does not avoid the exponential enumeration of values for the free variables. We believe there has been no investigation into quantifier elimination procedures that can be practically used for this theory. 

Such procedures are needed, for instance, in the formal verification of cipher programs involving finite field arithmetic~\cite{SmithD08,using} and polynomial dynamical systems over finite fields that arise in systems biology~\cite{virus,bio2,poly}. Take the S2VD virus competition model~\cite{virus} as an example, which we study in detail in Section 6: The dynamics of the system is given by a set of polynomial equations over the field $F_4$. We can encode image computation and invariant analysis problems as quantified formulas, which are solvable using quantifier elimination. As is mentioned in~\cite{virus}, there exists no verification method suitable for such systems over general finite fields so far.

In this paper we give an algebraic quantifier elimination algorithm for $T_q$. The algorithm relies on strong Nullstellensatz and Gr\"obner basis methods. We analyze its theoretical complexity, and show its practical application. 

In Section 3, we exploit the strong Nullstellensatz over finite fields and properties of elimination ideals, to show that Gr\"obner basis computation gives a way of eliminating quantifiers in formulas of the form $\exists \vec x (\bigwedge_i \alpha_i)$, where the $\alpha_i$s are atomic formulas and $\exists \vec x$ is a quantifier block. We then show, in Section 4, that the DNF-expansion of formulas can be avoided by using standard ideal operations to ``flatten'' the formulas. Any quantifier-free formula can be transformed into conjunctions of atomic formulas at the cost of introducing existentially quantified variables. This transformation is linear in the size of the formula, and can be seen as a generalization of the {\em Tseitin transformation}. Combining the techniques, we obtain a complete quantifier elimination algorithm. 

In Section 5, we analyze the complexity of our algorithm, which depends on the complexity of Gr\"obner basis computation over finite fields. For ideals in $F_q[\vec x]$ that contain $x_i^q-x_i$ for each $x_i$, Buchberger's Algorithm computes Gr\"obner bases within exponential time and space~\cite{Lakshman}. Using this result, the worst-case time/space complexity of our algorithm is bounded by $q^{O(|\varphi|)}$ when $\varphi$ contains no more than two {\em alternating blocks} of quantifiers, and $q^{q^{O(|\varphi|)}}$ for more alternations. Recently a polynomial-space algorithm for Gr\"obner basis computation over finite fields has been proposed in \cite{pspaceGB}, but it remains theoretical so far. If the new algorithm can be practically used, the worst-case complexity of quantifier elimination is $q^{O(|\varphi|)}$ for arbitrary alternations. 

Note that this seemingly high worst-case complexity, as is common for Gr\"obner basis methods, does not prevent the algorithm from being useful on practical problems. This is crucially different from the naive algorithm, which always requires exponential cost, not just in worst cases. In Section 6, we show how the algorithm is successfully applied in the analysis of a controller design in the S2VD virus competition model~\cite{virus}, which is a polynomial dynamical system over finite fields. The authors developed control strategies to ensure a safety property in the model, and used simulations to conclude that the controller is effective. However, using the quantifier elimination algorithm, we found bugs that show inconsistency between specifications of the system and its formal model. This shows how our algorithm can provide a practical way of extending formal verification techniques to models over finite fields.

Throughout the paper, omitted proofs are provided in the Appendix.

\section{Preliminaries}
\subsection{Ideals, Varieties, Nullstellensatz, and Gr\"obner Bases}

Let $k$ be any field and $k[x_1,...,x_n]$ the polynomial ring over $k$ with indeterminates $x_1,...,x_n$. An {\em ideal} generated by $f_1,...,f_m\in k[x_1,...,x_n]$ is $\langle f_1,...,f_m\rangle=\{h:h=\sum_{i=1}^m g_if_i$, $g_i\in k[x_1,...,x_n]\}.$ Let $\vec a \in k^n$ be an arbitrary point, and $f\in k[x_1,...,x_n]$ be a polynomial. We say that $f$ {\em vanishes} on $\vec a$ if $f(\vec a) = 0$. 
\begin{definition}
For any subset $J$ of $k[x_1,...,x_n]$, the {\bf affine variety} of $J$ over $k$ is {$V_n(J)=\{\vec  a\in k^n:\forall f\in J, f(\vec  a)=0\}.$}
\end{definition}
\begin{definition}
For any subset $V$ of $k^n$, the {\bf vanishing ideal}
of $V$ is defined as $I(V)=\{f\in k[x_1,...,x_n]: \forall \vec  a\in V, f(\vec 
a)=0\}.$
\end{definition}
\begin{definition}
Let $J$ be any ideal in $k[x_1,...,x_n]$, the {\bf radical} of $J$ is defined as $\sqrt J=\{f\in k[x_1,...,x_n]:\exists m\in \mathbb{N}, f^m\in J\}.$
\end{definition}
 When $J=\sqrt J$, we say $J$ is a radical ideal. The celebrated Hilbert Nullstellensatz established the correspondence between radical ideals and varieties:
\begin{theorem}[Strong Nullstellensatz~\cite{Lang}]\label{sN}
For an arbitrary field $k$, let $J$ be an ideal in $k[x_1,...,x_n]$. We have $I(V^a(J))=\sqrt J,$ where $k^a$ is the algebraic closure of $k$ and $V^a(J)=\{\vec  a \in (k^a)^n: \forall f\in J, f(\vec  a)=0\}.$
\end{theorem}

The method of Gr\"obner bases was introduced by Buchberger~\cite{buchberger76} for the algorithmic solution of various fundamental problems in commutative algebra. For an ideal $\langle f_1,...,f_m\rangle$ in a polynomial ring, Gr\"obner basis computation transforms $f_1,...,f_m$ to a canonical representation $\langle g_1,...,g_s\rangle=\langle f_1,...,f_m\rangle$ that has many useful properties. Detailed treatment of the theory can be found in \cite{grobnerbook}. 

\begin{definition}
Let $T=\{x_1^{\alpha_1}\cdots x_n^{\alpha_n}: \alpha_i\in N\}$ be the set of monomials in $k[x_1,...,x_n]$. A {\bf monomial ordering} $\prec$ on $T$ is a well-ordering on T satisfying \\(1)
For any $t\in T$, $1\prec t$\\(2) For all $t_1, t_2, s\in T$, $t_1\prec t_2$ then $t_1\cdot s\prec t_2\cdot s$.
\end{definition}

We order the monomials appearing in any single polynomial $f\in k[x_1,...,x_n]$ with respect to $\prec$. We write $LM(f)$ to denote the {\em leading monomial} in $f$ (the maximal monomial under $\prec$), and $LT(f)$ to denote the {\em leading term} of $f$ ($LM(f)$ multiplied by its coefficient). We write $LM(S)=\{LM(f):f\in S\}$ where $S$ is a set of polynomials.

Let $J$ be an ideal in $k[x_1,...,x_n]$. Fix any monomial order on $T$. The ideal of leading
monomials of $J$, $\langle LM(J)\rangle$, is the ideal generated by the leading monomials of all polynomials in $J$. Now we are ready to define:
\begin{definition}[Gr\"obner Basis~\cite{grobnerbook}]
A {\bf Gr\"obner basis} for $J$ is a set $GB(J)=\{g_1,...,g_s\}\subseteq J$ satisfying $\langle LM(GB(J)) \rangle=\langle
LM(J)\rangle.$ 
\end{definition}

\subsection{The First-order Theory over a Finite Field}

Let $F_q$ be an arbitrary finite field of size $q$, where $q$ is a prime power. We fix the structure to be $M_q=\langle F_q, 0, 1, +, \times\rangle$ and the signature $\mathcal{L}_q=\langle 0,1, +,\times \rangle$ (``$=$'' is a logical predicate). For quantified formulas, we write $\varphi(\vec x; \vec y)$ to emphasize that the $\vec x$ is a vector of quantified variables and $\vec y$ is a vector of free variables.

The standard first-order theory for each $M_q$ consists of the usual axioms for fields \cite{modelbook} plus $\exists x_1\cdots\exists x_q ((\bigwedge_{1\leq i< j \leq q} x_i \neq x_j )\wedge \forall y(\bigvee_i y = x_i))$, which fixes the size of the domain. We write this theory as $T_q$. In $\mathcal{L}_q$, we consider all the atomic formulas as polynomial equations $f=0$. The {\em realization} of a formula is the set of assignments to its free variables that makes the formula true over $M_q$. Formally:

\begin{definition}[Realization]
Let $\varphi(x_1,...,x_n)$ be a formula with free variables $\vec x = (x_1,...,x_n)$. The realization of $\varphi$, written as $\llbracket \varphi \rrbracket \subseteq F_q^n$, is inductively defined as:
\begin{itemize}
\item $\llbracket p = 0\rrbracket =_{df} V(\langle p\rangle) \subseteq F_q^n$ (in particular, $\llbracket \top\rrbracket = F_q^n$)

\item $\llbracket \neg \psi \rrbracket = F_q^n\setminus \llbracket \psi\rrbracket$

\item $\llbracket \psi_1\wedge \psi_2 \rrbracket = \llbracket \psi_1\rrbracket \cap \llbracket \psi_2\rrbracket$

\item $\llbracket\exists x_0. \psi(x_0,\vec x) \rrbracket = \{\langle a_1,...,a_n\rangle \in F_q^n: \exists a_0\in F_q, \mbox{ such that } \langle a_0,...,a_n\rangle \in \llbracket\psi\rrbracket \}$
\end{itemize}
\end{definition}
\begin{proposition}[Fermat's Little Theorem]
Let $F_q$ be a finite field. For any $a\in F_q$, we have $a^q-a=0$. Conversely, $V(x^q-x) = \llbracket x^q-x\rrbracket = F_q$.
\end{proposition}
\begin{definition}[Quantifier Elimination]
$T_q$ admits quantifier elimination if for any formula $\varphi(\vec x;\vec y)$, where the $\vec x$ variables are quantified and the $\vec y$ variables free, there exists a quantifier-free formula $\psi(\vec y)$ such that $\llbracket\varphi(\vec x; \vec y)\rrbracket = \llbracket \psi(\vec y) \rrbracket$. 
\end{definition}

\subsection{Nullstellensatz in Finite Fields}

The strong Nullstellensatz admits a special form over finite fields. This was proved for prime fields in \cite{germ} and used in~\cite{poly,Marchand98onthe}. Here we give a short proof that the special form holds over arbitrary finite fields, as a corollary of Theorem~\ref{sN}.

\begin{lemma}\label{first-lemma}
For any ideal $J\subseteq F_q[x_1,...,x_n]$, $J+ \langle
x_1^q-x_1,...,x_n^q-x_n\rangle$ is radical.
\end{lemma}

\begin{theorem}[Strong Nullstellensatz in Finite Fields]
\label{null}
For an arbitrary finite field $F_q$, let $J\subseteq F_q[x_1,...,x_n]$ be an ideal,
then $$I(V(J))=J+\langle x_1^q-x_1,...,x_n^q-x_n\rangle.$$
\end{theorem}
\begin{proof}
Apply Theorem~\ref{sN} to $J+\langle x_1^q-x_1,...,x_n^q-x_n\rangle$ and use Lemma~\ref{first-lemma}. We have 
$I(V^a(J+\langle x_1^q-x_1,...,x_n^q-x_n\rangle))= J+\langle x_1^q-x_1,...,x_n^q-x_n\rangle
$.
But since $V^a(\langle x_1^q-x_1,...,x_n^q-x_n\rangle)=F_q^n$, it follows that $$V^a(J+\langle
x_1^q-x_1,...,x_n^q-x_n\rangle)=V^a(J)\cap F_q^n = V(J).$$ Thus we obtain $I(V(J))=J+\langle x_1^q-x_1,...,x_n^q-x_n\rangle.$\qed
\end{proof}

\section{Quantifier Elimination Using Gr\"obner Bases}

In this section, we show that the key step in quantifier elimination can be realized by Gr\"obner basis computation. Namely, for any formula $\varphi$ of the form $\exists \vec x \bigwedge_{i=1}^r f_i (\vec x, \vec y) =0$, we can compute a quantifier-free formula $\psi(\vec y)$ such that $\llbracket \varphi(\vec x;\vec y) \rrbracket = \llbracket \psi(\vec y) \rrbracket$. We use the following notational conventions:
\begin{itemize}
\item $|\vec x|=n$ is the number of quantified variables and $|\vec y|=m$ the number of free variables. We write $\vec x^q-\vec x =_{df} \{x_1^q-x_1,...,x_n^q-x_n\}$ and $\vec y^q-\vec y =_{df} \{y_1^q-y_1,...,y_m^q-y_m\}$, and call them {\em field polynomials} (following \cite{germ}).

\item We use $\vec a = (a_1,...,a_n)\in F_q^n$ to denote the assignment for the $\vec x$ variables, and $\vec b=(b_1,...,b_m) \in F_q^m$ for the $\vec y$ variables. $(\vec a,\vec b)\in F_q^{n+m}$ is a complete assignment for all the variables in $\varphi$.

\item When we write $J\subseteq F_q[\vec x, \vec y]$ or a formula $\varphi(\vec x; \vec y)$, we assume that all the $\vec x, \vec y$ variables do occur in $J$ or $\varphi$. We assume that the $\vec x$ variables always rank higher than the $\vec y$ variables in the lexicographic order. 
\end{itemize}

\subsection{Existential Quantification and Elimination Ideals}

First, we show that eliminating the $\vec x$ variables is equivalent to {\em projecting} the variety $V(\langle f_1,...,f_r\rangle)$ from $F_q^{n+m}$ to $F_q^{m}$. 

\begin{lemma}\label{formula-ideal}
For $f_1,...,f_r\in F_q[\vec x, \vec y]$, we have $\llbracket \bigwedge_{i=1}^r f_i = 0 \rrbracket = V(\langle f_1,...,f_r\rangle)$.
\end{lemma}

\begin{definition}[Projection]
The $l$-th projection mapping is defined as: $$\pi_l: F_q^{N}\rightarrow F_q^{N-l},\pi_l((c_1,...,c_N))= ( c_{l+1},...,c_{N})$$ where $l<N$.
For any set $A\subseteq F_q^{N}$, we write $\pi_l(A)= \{\pi_i(\vec c): \vec c\in A\}\subseteq F_q^{N-l}.$
\end{definition}

\begin{lemma}\label{just-projection}
$\llbracket \exists \vec x \varphi (\vec x; \vec y)\rrbracket = \pi_n (\llbracket\varphi(\vec x; \vec y)\rrbracket)$.
\end{lemma}

Next, we show that the projection $\pi_n$ of the variety $V_{n+m}(\langle f_1,...,f_r\rangle)$ from $F_q^{n+m}$ to $F_q^m$,  is exactly the variety $V_m(\langle f_1,...,f_r\rangle\cap F_q[\vec y])$. 

\begin{definition}[Elimination Ideal \cite{ideals}]
Let $J\subseteq F_q[x_1,...,x_n]$ be an ideal. The l-th {\bf elimination ideal} $J_l$, for $1\leq l \leq N$, is the ideal of $F_q[x_{l+1},...,x_N]$ defined by $J_l = J\cap F_q[x_{l+1},...,x_N].$
\end{definition}

The following lemma shows that adding field polynomials does not change the realization. For $f_1,...,f_r\in F_q[\vec x, \vec y]$, we have:
\begin{lemma} \label{extend-formulas}
$\llbracket\bigwedge_{i=1}^r f_i = 0\rrbracket= \llbracket \bigwedge_{i=1}^r f_i=0 \wedge \bigwedge (x_i^q-x_i=0)\wedge \bigwedge (y_i^q-y_i=0) \rrbracket.$
\end{lemma}

Now we can prove the key equivalence between projection operations and elimination ideals. This requires the use of Nullstellensatz for finite fields.

\begin{theorem}\label{key-lemma}
Let $J\subseteq F_q[\vec x, \vec y]$ be an ideal which contains the field polynomials for all the variables in $J$. We have $\pi_n(V(J)) = V(J_n).$
\end{theorem}
\begin{proof} We show inclusion in both directions.
\begin{itemize}
\item $\pi_n(V(J)) \subseteq V(J_n):$ 

For any $\vec b\in \pi_n(V(J))$, there exists $\vec a\in F_q^n$ such that $(\vec a, \vec b)\in V(J)$. That is, $(\vec a, \vec b)$ satisfies all polynomials in $J$; in particular, $\vec b$ satisfies all polynomials in $J$ that only contain the $\vec y$ variables ($\vec a$ is not assigned to variables). Thus, $\vec b \in V(J\cap F_q[\vec y]) = V(J_n).$

\item $V(J_n)\subseteq \pi_n(V(J)):$ 

Let $\vec b$ be a point in $F_q^m$ such that $\vec b\not\in \pi_n(V(J))$. Consider the polynomial $$f_{\vec b} = \prod_{i=1}^{m} (\prod_{c\in F_q\setminus\{b_i\}} (y_i - c)).$$

$f_{\vec b}$ vanishes on all the points in $F_q^n$, except $\vec b = ( b_1,...,b_m)$, since $(y_i-b_i)$ is excluded in the product for all $i$. In particular, $f_{\vec b}$ vanishes on all the points in $V(J)$, because for each $(\vec a, \vec b')\in V(J)$, $\vec b'$ must be different from $\vec b$, and $f_{\vec b}(\vec a, \vec b') = f_{\vec b}(\vec b') = 0$ (since there are no $\vec x$ variables). Therefore, $f_{\vec b}$ is contained in the vanishing ideal of $V(J)$, i.e., $f_{\vec b}\in I(V(J))$.

Now, Theorem \ref{null} shows $I(V(J))= J+\langle \vec x^q-\vec x, \vec y^q-\vec y\rangle$.
Since $J$ already contains the field polynomials, we know $J+\langle \vec x^q-\vec x, \vec y^q-\vec y\rangle = J$, and consequently $I(V(J))= J.$ Since $f_{\vec b}\in I(V(J))$, we must have $f_{\vec b}\in J$. But on the other hand, $f_{\vec b}\in F_q[\vec y]$. Hence $f_{\vec b}\in J\cap F_q[\vec y] = J_n$. But since $f_{\vec b}(\vec b)\neq 0$,  we know $\vec b\not\in V(J_n)$.\qed
\end{itemize}
\end{proof}

\subsection{Quantifier Elimination using Elimination Ideals}

Theorem \ref{key-lemma} shows that to obtain the projection of a variety over $F_q$, we only need to take the variety of the corresponding elimination ideal. In fact, this can be easily done using the Gr\"obner basis of the original ideal:

\begin{proposition}[cf. \cite{ideals}]\label{GB}
Let $J\subseteq F_q[x_1,...,x_N]$ be an ideal and let $G$ be the Gr\"obner basis of $J$ with respect to the lexicographic order $x_1\succ\cdots\succ x_N$. Then for every $1\leq l\leq N$, $G\cap F_q[x_{l+1},...,x_N]$ is a Gr\"obner basis of the $l$-th elimination ideal $J_l$. That is, $J_l=\langle G\rangle \cap F_q[x_{l+1},...,x_N]=\langle G\cap F_q[x_{l+1},...,x_N]\rangle.$
\end{proposition}

Now, putting all the lemmas together, we arrive at the following theorem:

\begin{theorem}\label{main-theorem}
Let $\varphi(\vec x;\vec y)$ be $\exists \vec x.(\bigwedge_{i=1}^r f_i=0)$ be a formula in $\mathcal{L}_q$, with $f_i\in F_q[\vec x, \vec y]$. Let $G$ be the Gr\"obner basis of $\langle f_1,...,f_r, \vec x^q-\vec x, \vec y^q-\vec y \rangle$. Suppose $G\cap F_q[\vec y] = \{g_1,...,g_s\},$ then we have $\llbracket \varphi \rrbracket = \llbracket \bigwedge_{i=1}^s (g_i=0)\rrbracket.$
\end{theorem}

\begin{proof}
We write $J=\langle f_1,...,f_r,\vec x^q-\vec x, \vec y^q-\vec y\rangle$ for convenience. First, by Lemma \ref{extend-formulas}, adding the polynomials $\vec x^q-\vec x$ and $\vec y^q -\vec y$ does not change the realization:
{$$\llbracket \varphi\rrbracket = \llbracket\exists \vec x.(\bigwedge_{i=1}^r f_i=0)\rrbracket = \llbracket\exists \vec x.(\bigwedge_{i=1}^r f_i=0\wedge \bigwedge_{i=1}^n( x_i^q- x_i = 0) \wedge \bigwedge_{i=1}^m (y_i^q- y_i=0))\rrbracket $$}Next, by Lemma \ref{just-projection}, the quantification on $\vec x$ corresponds to projecting a variety:
{$$\llbracket\exists \vec x.(\bigwedge_{i=1}^r f_i=0\wedge \bigwedge_{i=1}^n( x_i^q- x_i = 0) \wedge \bigwedge_{i=1}^m (y_i^q- y_i=0))\rrbracket= \pi_n(V(J)).$$}Using Theorem \ref{key-lemma}, we know that the projection of a variety is equivalent to the variety of the corresponding elimination ideal, i.e., $\pi_n(V(J)) =V(J\cap F_q[\vec y])$. Now, using the property of Gr\"obner bases in Proposition \ref{GB}, we know the elimination ideal $\langle G\rangle\cap F_q[\vec y]$ is generated by $G\cap F_q[\vec y]$:
{$$V(J\cap F_q[\vec y])=V(\langle G\rangle \cap F_q[\vec y]) = V(\langle G\cap F_q[\vec y] \rangle)= V(\langle g_1,...,g_s\rangle)$$}Finally, by Lemma \ref{formula-ideal}, an ideal is equivalent to the conjunction of atomic formulas given by the generators of the ideal: $V(\langle g_1,...,g_s\rangle)= \llbracket \bigwedge_{i=1}^s g_i=0\rrbracket.$

Connecting all the equations above, we have shown $\llbracket \varphi \rrbracket = \llbracket \bigwedge_{i=1}^s g_i=0\rrbracket.$ Note that $g_1,...,g_s\in F_q[\vec y]$ (they do not contain $\vec x$ variables).\qed
\end{proof}

\section{Formula Flattening with Ideal Operations}

If negations on atomic formulas can be eliminated (to be shown in Lemma \ref{no-negation}), Theorem \ref{main-theorem} already gives a direct quantifier elimination algorithm. That is, we can always use duality to make the innermost quantifier block an existential one, and expand the quantifier-free part to DNF. Then the existential block can be distributed over the disjuncts and Theorem \ref{main-theorem} is applied. However, this direct algorithm {\em always} requires exponential blow-up in expanding formulas into DNF. 

We show that the DNF-expansion can be avoided: Any quantifier-free formula can be transformed into an equivalent formula of the form $\exists \vec z.(\bigwedge_{i=1}^r f_i=0)$, where $\vec z$ are new variables and $f_i$s are polynomials. The key is that Boolean conjunctions and disjunctions can both be turned into additions of ideals; in the latter case new variables need be introduced. This transformation can be done in linear time and space, and is a generalization of the {\em Tseitin transformation} from $F_2$ to general finite fields. 

We use the usual definition of ideal addition and multiplication. Let $J_1=\langle f_1,...,f_r\rangle$ and $J_2=\langle g_1,...,g_s\rangle$ be ideals, and $h$ be a polynomial. Then $J_1+J_2=\langle f_1,...,f_, g_1,...,g_s\rangle$ and $J_1\cdot h = \langle f_1\cdot h,...,f_r\cdot h\rangle$.

\begin{lemma}[Elimination of Negations]
Suppose $\varphi$ is a quantifier free formula in $\mathcal{L}_q$ in NNF and contains $k$ negative atomic formulas. Then there is a formula $\exists \vec z.\psi$, where $\psi$ contains new variables $\vec z$ but no negative atoms, such that $\llbracket \varphi \rrbracket = \llbracket \exists \vec z.\psi \rrbracket $.
\label{no-negation}\end{lemma}

\begin{lemma}[Elimination of Disjunctions]\label{no-disjunction}
Suppose $\psi_1$ and $\psi_2$ are two formulas in variables $x_1,...,x_n$, and $J_1$ and $J_2$ are ideals in $F_q[x_1,...,x_n]$ satisfying $\llbracket \psi_1\rrbracket = V(J_1)$ and $\llbracket \psi_2 \rrbracket = V(J_2)$. Then, using $x_0$ as a new variable, we have $\llbracket \psi_1\vee \psi_2 \rrbracket = V(J_1)\cup V(J_2) =\pi_0( V(x_0 J_1+(1-x_0)J_2)).$
\end{lemma}


\begin{theorem}
For any quantifier-free formula $\varphi(\vec x)$ given in NNF, there exists a formula $\psi$ of the form $\exists \vec u, \vec v (\bigwedge_i( f_i(\vec x, \vec u, \vec v)=0))$ such that $\llbracket \varphi \rrbracket = \llbracket \psi \rrbracket$. Furthermore, $\psi$ can be generated in time $O(|\varphi|)$, and also $|\vec u|+|\vec v|=O(|\varphi|)$.
\label{flatten-theorem}
\end{theorem}
\begin{proof}
Since $\varphi(\vec x)$ is in NNF, all the negations occur in front of atomic formulas. We first use Lemma \ref{no-negation} to eliminate the negations. Suppose there are $k$ negative atomic formulas in $\varphi$, we obtain $\llbracket \varphi \rrbracket= \llbracket\exists u_1,...,u_k.\varphi'\rrbracket$. Now $\varphi'$ does not contain negations.

We then prove that there exists an ideal $J_{\varphi'}$ for $\varphi'$ satisfying $\pi_{|\vec v|}(V(J_{\varphi'}))=\llbracket \varphi'\rrbracket$, where $\vec v$ are the introduced variables (which rank higher than the existing variables in the variable ordering, so that the projection $\pi_{|\vec v|}$ truncates assignments on the $\vec v$ variables).
\begin{itemize}
\item If $\varphi'$ is an atomic formula $f=0$, then $J_{\varphi'}=\langle f\rangle$;

\item If $\varphi'$ is of the form $\theta_1\wedge \theta_2$, then $J_{\varphi'}= J_{\theta_1}+J_{\theta_2}$;

\item If $\varphi'$ is of the form $\theta_1\vee \theta_2$, then $J_{\varphi'}= v_i\cdot J_{\theta_1}+(1-v_i)\cdot J_{\theta_2}$, where $v_i$ is new.
\end{itemize}
Note that the new variables are only introduced in the disjunction case, and therefore the number of $\vec v$ variables equals the number of disjunctions. Following Lemma \ref{formula-ideal} and \ref{no-disjunction}, the transformation preserves the realization of the formula in each case. Hence, we have $\pi_{\vec v}(V(J_{\varphi'}))=\llbracket \varphi'\rrbracket$. Writing $J_{\varphi'}=\langle f_1,...,f_r\rangle$, we know $\llbracket\varphi\rrbracket = \llbracket\exists \vec u.\varphi'\rrbracket=\llbracket\exists \vec u\exists\vec v.\bigwedge_{i=1}^r f_i\rrbracket.$ Notice that the number of rewriting steps is bounded by the number of logical symbols appearing in $\varphi$. Hence the transformation is done in time linear in the size of the formula. The number of new variables is equal to the number of negations and disjunctions. 
\qed
\end{proof}

\section{Algorithm Description and Complexity Analysis}

We now describe the full algorithm using the following notations:
\begin{itemize}
\item The input formula is given by $\varphi = Q_1\vec x_1\cdots Q_m \vec x_m \psi $. Each $Q_i\vec x_i$ represents a {\em quantifier block}, where $Q_i$ is either $\exists$ or $\forall$. $Q_{i}$ and $Q_{i+1}$ are different quantifiers. We write $\vec x=(\vec x_1,...,\vec x_m)$. $\psi$ is a quantifier-free formula in $\vec x$ and $\vec y$ given in NNF, where $\vec y$ are free variables. 
\item We assume the innermost quantifier is existential, $Q_m=\exists$. (Otherwise we apply quantifier elimination on the negation of the formula.)
\end{itemize}
\subsection{Algorithm Description}
\begin{algorithm}[h]
\caption{Quantifier Elimination for $\varphi= Q_1 \vec x_1\cdots Q_m\vec x_m. \psi$}
\label{mainalgo}
    \begin{algorithmic}[1]
\STATE {\bf Input:} $\varphi= Q_1 \vec x_1\cdots Q_m\vec x_m.\psi(\vec x_1,...,\vec x_m, \vec y)$ where $m$ is the number of quantifier alternations, $Q_mx_m$ is an existential block ($Q_m = \exists$), and $\psi$ is in negation normal form.
\STATE {\bf Output:} A quantifier-free equivalent formula of $\varphi$
\STATE {\bf \em Procedure} {\bf\em QE($\varphi$)}
\WHILE{$m\geq 1$}
\STATE $\exists\vec u. \psi' \gets$ {\bf\em Eliminate\_Negations($\psi$)}
\STATE $\exists\vec v. (f_1=0\wedge \cdots \wedge f_r=0)\gets$ {\bf\em Formula\_Flattening($\psi'$)}
\STATE $\varphi \gets Q_1 \vec x_1\cdots Q_m\vec x_m\exists \vec u\exists \vec v. (f_1=0\wedge \cdots \wedge f_r=0)$
\STATE $\{g_1,...,g_{s}\}$ = Gr\"obner\_Basis($\langle f_1,...,f_r,\vec x^q-\vec x, \vec u^q-\vec u, \vec v^q-\vec v\rangle$)
\IF{$m=1$}
\STATE $\varphi\gets g_1=0\wedge \cdots \wedge g_s=0$
\STATE {\bf break}
\ENDIF
\STATE $\varphi \gets Q_1\vec x_1\cdots  Q_{m-2} \vec x_{m-2} Q_{m-1} \vec x_{m-1}. (\bigwedge_{i=1}^s g_i=0)$ where $Q_{m-1}=\forall$
\STATE $\varphi \gets Q_1 \vec x_1\cdots Q_{m-2} \vec x_{m-2}. (\bigwedge_{i=1}^{s} \neg \exists \vec x_{m-1} (g_i\neq 0))$ 
\FOR{$i=1$ to $s$}
\STATE $\bigwedge_{j=1}^{t_i} h_{ij}=0 \gets${\bf \em QE}$(\exists \vec x_{m-1} (g_i\neq 0))$
\ENDFOR
\STATE $\varphi \gets Q_1 \vec x_1\cdots Q_{m-2} \vec x_{m-2} \bigwedge_{i=1}^{s} (\bigvee_{j=1}^{t_i} h_{ij}\neq 0)$
\STATE $m\gets m-2$
\ENDWHILE
\RETURN $\varphi$
    \end{algorithmic}
\end{algorithm}

Section 3 shows how to eliminate existential quantifiers over conjunctions of positive atomic formulas. Section 4 shows how formulas can be put into conjunctions of positive atoms with new quantified variables. It follows that we can always eliminate the innermost existential quantifiers, and iterate the process by flipping the universal quantifiers into existential ones. We first emphasize some special features of the algorithm:
\begin{itemize}
\item In each elimination step, a full {\em quantifier block} is eliminated. This is desirable in practical problems, which usually contain many variables but few alternating quantifier blocks. For instance, many verification problems are expressible using two blocks of quantifiers ($\forall\exists$-formulas). 
\item The quantifier elimination step essentially transforms an ideal to another ideal. This corresponds to transforming conjunctions of atomic formulas to conjunctions of new atomic formulas. Therefore, the quantifier elimination steps do not introduce new nesting of Boolean operators.
\item The algorithm always directly outputs CNF formulas.
\end{itemize}

A formal description of the full algorithm is given in Algorithm \ref{mainalgo}. The main steps in the algorithm are explained below. Each loop of the algorithm contains three main steps. In Step 1, $\varphi$ is flattened; in Step 2, the innermost existential quantifier block is eliminated; in Step 3, the next (universal) quantifier block is eliminated and the process loops back to Step 1. The algorithm terminates either after Step 2 or Step 3, when there are no remaining quantifiers to be eliminated.

$\bullet$ {\bf Step 1: (Line 5-7)} 

First, since $\psi$ is in NNF, we use Theorem \ref{flatten-theorem} to eliminate the negations and disjunctions in $\psi$ to get $\llbracket\varphi\rrbracket=\llbracket Q_1\vec x_1\cdots Q_m\vec x_m \exists \vec u\exists \vec v. (\bigwedge_{i=1}^{r} f_i=0)\rrbracket$, where $\vec u$ are the variables introduced for eliminating negations (Lemma \ref{no-negation}), and $\vec v$ are the variables introduced for eliminating disjunctions (Lemma \ref{no-disjunction}).

$\bullet$ {\bf Step 2: (Line 8-12)}

Since $Q_m=\exists$, using Theorem 4.1, we can eliminate the variables $\vec x_m, \vec u, \vec v$ simultaneously by computing $$\{g_1,...,g_{r_1}\} = GB(\langle f_1,...,f_r, \vec x_m^q-\vec x_m, \vec u^q-\vec u, \vec v^q-\vec v, \vec y^q-\vec y\rangle)\cap F_q[\vec x_1,...,\vec x_{m-1}, \vec y].$$ 

Now we have $\llbracket \varphi \rrbracket = \llbracket Q_1\vec x_1\cdots Q_{m-1}\vec x_{m-1}.(\bigwedge_{i=1}^{s} (g_i=0) ) \rrbracket.$

If there are no more quantifiers, the output is $\bigwedge_{i=1}^{s} (g_i=0)$, which is in CNF.

$\bullet$ {\bf Step 3: (Line 13-18)} 

Since $Q_{m-1}=\forall$, we distribute the block $Q_{m-1}\vec x_{m-1}$ over the conjuncts:
{$$\llbracket \varphi\rrbracket = \llbracket Q_1\vec x_1\cdots Q_{m-2}\vec x_{m-2} (\bigwedge_{i=1}^{s} (\neg\exists\vec x_{m-1}\neg (g_i=0)) )\rrbracket 
$$}Now we do elimination recursively on $\exists \vec x_{m-1}(\neg g_i=0)$ for each $i\in \{1,...,s\}$, which can be done using only Step 1 and Step 2. We obtain:{
\begin{eqnarray}\label{secondform}
\llbracket \exists  \vec x_{m-1} (\neg g_i=0)\rrbracket = \llbracket \exists \vec x_{m-1}\exists u'. (g_i\cdot u' -1 = 0)\rrbracket = \llbracket \bigwedge_{j=1}^{t_i} h_{ij}=0 \rrbracket
\end{eqnarray}
and the formula becomes (note that the extra negation is distributed){
\begin{eqnarray}\label{firstform}
\llbracket \varphi \rrbracket = \llbracket Q_1\vec x_1\cdots Q_{m-2}\vec x_{m-2}.(\bigwedge_{i=1}^{s} (\bigvee_{j=1}^{t_i} h_{ij}\neq 0))\rrbracket.
\end{eqnarray}
If there are no more quantifiers left, the output formula is $\bigwedge_{i=1}^{s} (\bigvee_{j=1}^{t_i} h_{ij}\neq 0)$, which is in CNF. Otherwise, $Q_{m-2} = \exists$, and we return to Step 1.

\begin{theorem}[Correctness]
Let $\varphi(\vec x; \vec y)$ be a formula $Q_1\vec x_i\cdots Q_m \vec x_m.\psi$ where $Q_m=\exists$ and $\psi$ is in NNF. Algorithm \ref{mainalgo} computes a quantifier-free formula $\varphi'(\vec y)$, such that $\llbracket \varphi(\vec x; \vec y)\rrbracket = \llbracket \varphi'(\vec y) \rrbracket$ and $\varphi'$ is in CNF.
\end{theorem}

\subsection{Complexity Analysis}

The worst-case complexity of Gr\"obner basis computation on ideals in $F_q[\vec x]$ that contain $x_i^q-x_i$ for each variable $x_i$ is known to be single exponential in the number of variables in time and space. This follows from the complexity result for Gr\"obner basis computation of zero-dimensional radical ideals~\cite{Lakshman} (a direct proof can be found in \cite{gao09}). 

\begin{proposition}\label{gb-time}
Let $J=\langle f_1,...,f_r, \vec x^q-\vec x\rangle\subseteq F_q[x_1,...,x_n]$ be an ideal. The time and space complexity of Buchberger's Algorithm is bounded by $q^{O(n)}$, assuming that the length of input ($f_1,...,f_r$) is dominated by $q^{O(n)}$.
\end{proposition}

Now we are ready to estimate the complexity of our algorithm.

\begin{theorem}[Complexity]\label{complexity}
Let $\varphi$ be the input formula with $m$ quantifier blocks. When $m\leq 2$, the time/space complexity of Algorithm 1 is bounded by $q^{O(|\varphi|)}$. Otherwise, it is bounded by $q^{q^{O(|\varphi|)}}$.
\end{theorem}
\begin{proof}
The complexity is dominated by Gr\"obner basis computation, whose complexity is determined by the number of variables occurring in the ideal. When $m\leq 2$, the main loop is executed once, and the number of newly introduced variables is bounded by the original length of the input formula. Therefore, Gr\"obner basis computations can be done in single exponential time/space. When $m>2$, the number of newly introduced variables is bounded by the length of the formula obtained from the previous run of the main loop, which can itself be exponential in the number of the remaining variables.  In that case, Gr\"obner basis computation can take double exponential time/space. 

$\bullet$ {\bf Case $m\leq 2$:}

In Step 1, the number of the introduced $\vec u$ and $\vec v$ variables equals to the number of negations and disjunctions that appear in the $\varphi$. Hence the total number of variables is bounded by the length of $\varphi$. The flattening takes linear time and space, $O(|\varphi|)$, as proved in Theorem \ref{flatten-theorem}.

In Step 2, by Proposition \ref{gb-time}, Gr\"obner basis computation takes time/space $q^{O(|\varphi|)}$. 

In Step 3, the variables $\vec x_m, \vec u, \vec v$ have all been eliminated. The length of each $g_iu'-1$ (see Formula (\ref{secondform}) in Step 3) is bounded by the number of monomials consisting of the remaining variables, which is $O(q^{(|\vec y|+\sum_{i=1}^{m-1}|\vec x_i|)})$ (because the degree on each variable is lower than $q$). Following Proposition \ref{gb-time}, Gr\"obner basis computation on each $g_iu'-1$ takes time and space $q^{O(|\vec y|+\sum_{i=1}^{m-1}|\vec x_i|)}$, which is dominated by $q^{O(|\varphi|)}$. Also, since the number $s$ of conjuncts is the number of polynomials in the Gr\"obner basis computed in the previous step, we know $s$ is bounded by $q^{O(|\varphi|)}$. In sum, Step 3 takes $q^{O(|\varphi|)}$ time/space in worst case.

Thus, the algorithm has worst-case time and space complexity $q^{O(|\varphi|)}$ when $m\leq 2$.

$\bullet$ {\bf Case $m>2$:}

When $m> 2$, the main loop is iterated for more than one round. The key change in the second round is that, the initial number of conjunctions and disjunctions in each conjunct could both be exponential in the number of the remaining variables ($\vec x_1,...,\vec x_{m-2}$). That means, writing the max of $t_i$ as $t$ (see Formula~(\ref{firstform}) in Step 3), both $s$ and $t$ can be of order $q^{O(|\varphi|)}$. 

In Step 1 of the second round, the number of the $\vec u$ variables introduced for eliminating the negations is $s\cdot t$. The number of the $\vec v$ variables introduced for eliminating disjunctions is also $s\cdot t$. Hence the flattened formula may now contain $q^{O(|\varphi|)}$ variables. 

In Step 2 of the second round, Gr\"obner basis computation takes time and space exponential in the number of variables. Therefore, Step 2 can now take $q^{q^{O(|\varphi|)}}$ in time and space. 

In Step 3 of the second round, however, the number of conjuncts $s$ {\em does not} become doubly exponential. This is because $g_i$ in Step 3 no longer contains the exponentially many introduced variables -- they were already eliminated in the previous step. Thus $s$ is reduced back to single exponential in the number of the remaining variables; i.e., it is bounded by $q^{O(|\varphi|)}$. Similarly, the Gr\"obner basis computation on each $g_iu'-1$, which now contains variables $\vec x_1,...,\vec x_{m-1}, \vec y$,  takes time and space $q^{O(|\varphi|)}$. In all, Step 3 takes time and space $q^{O(|\varphi|)}$.

In sum, the second round of the main loop can take time/space $q^{q^{O(|\varphi|)}}$. But at the end of the loop, the size of formula is reduced to $q^{O(|\varphi|)}$ after the Gr\"obner basis computations, because it is at most single exponential in the number of the remaining variables. Therefore, the double exponential bound remains for future iterations of the main loop.\qed
\end{proof}

Recently, \cite{pspaceGB} reports a Gr\"obner basis computation algorithm in finite fields using polynomial space. This algorithm is theoretical and cannot be applied yet. Given the analysis above, if such a polynomial-space algorithm for Gr\"obner basis computation can be practically used, the intermediate expressions do not have the double-exponential blow-up. On the other hand, it does not lower the space bound of our algorithm to polynomial space, because during flattening of the disjunctions, the introduced terms are multiplied together. To expand the introduced terms, one may still use exponential space. It remains further work to investigate whether the algorithm can be practically used and how it compares with Buchberger's Algorithm.
\begin{proposition}
If there exists a polynomial-space Gr\"obner basis computation algorithm over finite fields for ideals containing the field polynomials, the time/space complexity of our algorithm is bounded by $q^{O(|\varphi|)}$.
 \label{pspace}
\end{proposition}

\section{Example and Application}

\subsection{A Walk-through Example}

Consider the following formula over $F_3$: $$\varphi: \exists b \forall a \exists y\exists x.((y=ax^2+bx+c)\wedge (y= ax))$$ which has three alternating quantifier blocks and one free variable. We ask for a quantifier-free formula $\psi(c)$ equivalent to $\varphi$. 

We fix the lexicographic ordering to be $x\succ y\succ a\succ b\succ c.$ First, we compute the Gr\"obner basis $G_0$ of the ideal: {$\langle y-ax^2-bx-c, y-ax, x^3-x,y^3-y,a^3-a,b^3-b,c^3-c\rangle,$}and obtain the Gr\"obner basis of the elimination ideal { $$G_1= G_0\cap F_3[a,b,c]=\{abc + ac^2 + b^2c - c, a^3-a, b^3 - b, c^3 - c\}.$$}After this, $x$ and $y$ have been eliminated, and we have:
 {\begin{eqnarray*}
\llbracket \varphi \rrbracket &= &\llbracket\exists b\forall a.((abc + ac^2 + b^2c - c=0)\wedge (a^3-a=0) \wedge (b^3 - b=0) \wedge (c^3 - c=0))\rrbracket\\
&=& \llbracket \exists b\forall a.(abc + ac^2 + b^2c - c=0)\rrbracket  \ \ \ \ \ \ \ \\
&=& \llbracket \exists b.(\neg \exists a \exists u.(u(abc + ac^2 + b^2c - c)-1=0))\rrbracket   
 \end{eqnarray*}}Now we eliminate quantifiers in $\exists a \exists u ((abc + ac^2 + b^2c - c)\cdot u-1=0),$ again by computing the Gr\"obner basis $G_2$ of the ideal {$$\langle (abc + ac^2 + b^2c - c)u-1,a^3-a,b^3-b,c^3-c,u^3-u\rangle \cap F_3[b,c].$$} We obtain $G_2=\{ b^2 - bc, c^2 - 1\}$. Therefore 
{$\llbracket \varphi \rrbracket = \llbracket \exists b (\neg (b^2 - bc=0\wedge c^2 - 1=0))\rrbracket$.} (Note that if both $b$ and $c$ are both free variables, $b^2 - bc\neq 0\vee c^2 - 1\neq 0$ would be the quantifier-free formula containing $b,c$ that is equivalent to $\varphi$.) 

Next, we introduce $u_1$ and $u_2$ to eliminate the negations, and $v$ to eliminate the disjunction: $$\llbracket\varphi\rrbracket = \llbracket\exists b \exists u_1\exists u_2\exists v.( ((b^2 - bc)u_1-1)v=0\wedge ((c^2-1)u_2)(1-v)=0)\rrbracket.$$

We now do a final step of computation of the Gr\"obner basis $G_3$ of:
{$$\langle ((b^2 - bc)u_1-1)v, ((c^2-1)u_2)(1-v),  b^3-b, c^3-c, u_1^3-t_1, u_2^3-t_2, v^3-v\rangle\cap F_3[c].$$}We obtain $G_3=\{c^3-c\}$. This gives us the result formula $\llbracket\varphi\rrbracket =\llbracket c^3-c=0\rrbracket,$ which means that $c$ can take any value in $F_3$ to make the formula true.
\subsection{Analyzing a Biological Controller Design}

We studied a virus competition model named S2VD~\cite{virus}, which models the dynamics of virus competition as a polynomial system over finite fields. The authors aimed to design a controller to ensure that one virus prevail in the environment. They pointed out that there was no existing method for verifying its correctness. The current design is confirmed effective by computer simulation and lab experiments for a wide range of initializations. We attempted to establish the correctness of the design with formal verification techniques. However, we found bugs in the design. 

All the Gr\"obner basis computations in this section are done using scripts in the SAGE system~\cite{sage}, which uses the underlying Singular implementation~\cite{singular}. All the formulas below are solved within 5 seconds on a Linux machine with 2GHz CPU and 2GB RAM. They involve around 20 variables over $F_4$, with nonlinear polynomials containing multiplicative products of up to 50 terms.  

\begin{figure}[ht]
\centering
\label{fig}
\includegraphics[width=80pt]{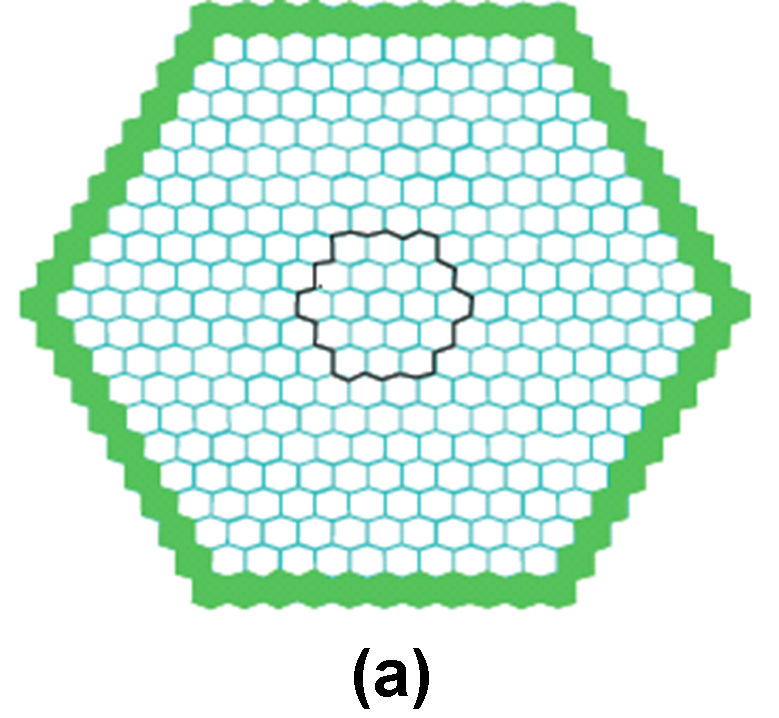}\hspace{.25in}
\includegraphics[width=70pt]{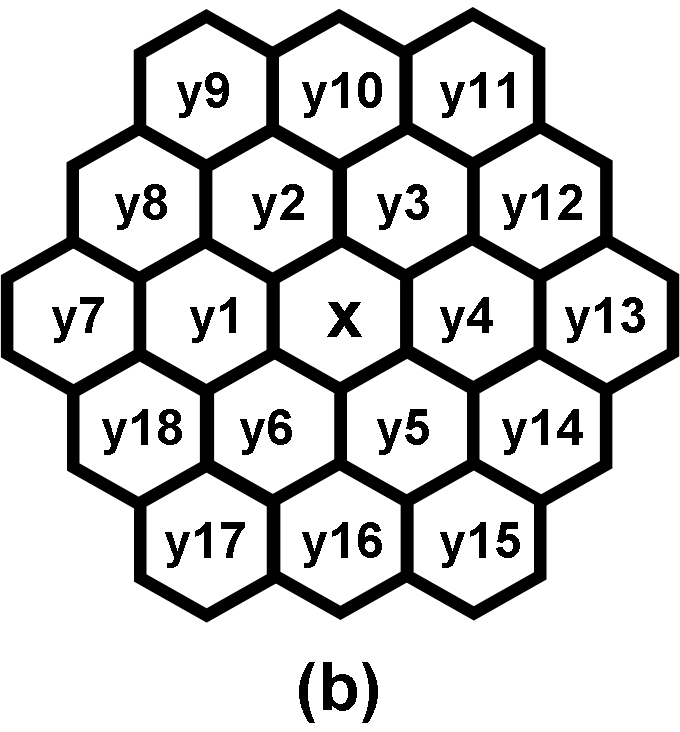}\hspace{.25in}
\includegraphics[width=50pt]{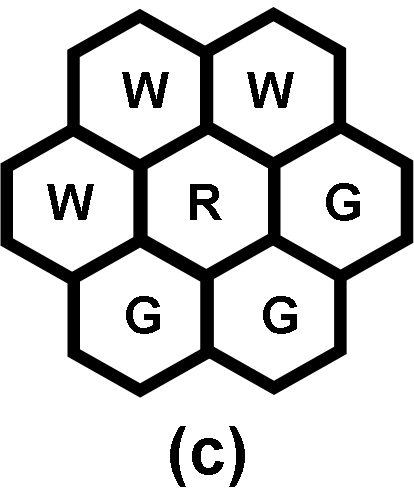}
\caption{(a) The ten rings of S2VD; (b) Cell x and its neighbor $\vec y$ cells; (c) The counterexample}
\end{figure}

\subsubsection{The S2VD Model}The model consists of a hexagonal grid of cells. Each hexagon represents a cell, and each cell has six neighbors. There are four possible colors for each cell. A green cell is infected with (the good) Virus G, and a red cell is infected with (the bad) Virus R. When the two viruses meet in one cell, Virus G captures Virus R and the cell becomes yellow. A cell not infected by any virus is white. The dynamics of the system is determined by the interaction of the viruses. 

There are ten rings of cells in the model, with a total of 331 cells (Figure 1(a)). In the initial configuration, the cells in Ring 4 to 10 are set to white, and the cells in Ring 1 to 3 can start with arbitrary colors. The aim is to have a controller that satisfies the following safety property: The cells in the outermost ring are either green or white at all times. The proposed controller detects if any cell has been infected by Virus R, and injects cells that are ``one or two rings away'' from it with Virus G. The injected Virus G is used to block the further expansion of Virus R.

Formally, the model is a polynomial system over the finite field $F_4=\{0,1,a,a+1\}$, with each element representing one color: $(0, green), (1, red), (a, white), (a+1, yellow)$. The dynamics is given by the function $f: F_4^{331}\rightarrow F_4^{331}.$ For each cell $x$, its dynamics $f_x$ is determined by the color of its six neighbors $y_1,...,y_6$, specified by the nonlinear polynomial $f_x=_{df}\gamma_2^2+ \gamma_2\gamma_1^3+a^2(\gamma_1^3+\gamma_1^2+\gamma_1)$, where $\gamma_1 = \sum_{i=1}^6 y_i$ and $\gamma_2 = \sum_{i\neq j} y_iy_j$. The designed controller is specified by another function $g:F_4^{331}\rightarrow F_4^{331}$: For each cell $x$, with $y_1,...,y_{18}$ representing the cells in the two rings surrounding it, we define $g_x=_{df}\prod_{i=1}^{18} (1-y_i)^3$. More details can be found in \cite{virus}.

\subsubsection{Applying Quantifier Elimination}
We first try checking whether the safety property itself forms an inductive invariant of the system (which is a strong sufficient check). To this end, we check whether the controlled dynamics of the system remain inside the invariant on the boundary (Ring 10) of the system. Let $x$ be a cell in Ring 10 and $\vec y = (y_1,...,y_{18})$ be the cells in its immediate two rings. We assume the cells outside Ring 10 ($y_8,...,y_{12},y_2,y_3$) are white. See Figure 1(b) for the coding of the cells. We need to decide the formula:
{\footnotesize\begin{eqnarray}\label{thirdform}\forall x(\underbrace{(\exists \vec y ((\bigwedge_{i=8}^{12} (y_i=a) \wedge y_2=a \wedge y_3=a)\wedge\mbox{Safe}(\vec y)\wedge x=F_x(\vec y)))}_{\varphi_1}\rightarrow \underbrace{x(x-a)=0)}_{\mbox{``green/white''}}\end{eqnarray}}
where (writing $\gamma_1=\sum_{i=1}^{6} y_i, \gamma_2 = \sum_{i\neq j\in \{1,...,6\}} y_iy_j$){\footnotesize\begin{eqnarray*}
\mbox{Safe}(\vec y)&=_{df}&(y_1(y_1-a) = 0\wedge y_4(y_4-a)=0\wedge y_7(y_7-a)=0\wedge y_{13}(y_{13}-a)=0)\\
F_x(\vec y)&=_{df}&(\gamma_2^2+ \gamma_2\gamma_1^3+a^2(\gamma_1^3+\gamma_1^2+\gamma_1))\cdot (\prod_{i=1}^{18}(1-y_i))^3
\end{eqnarray*}}After quantifier elimination, Formula (\ref{thirdform}) turns out to be false. In fact, we obtained $\llbracket \varphi_1\rrbracket = \llbracket x^4-x=0\rrbracket$. Therefore, the safety property itself is not an inductive invariant of the system. We realized that there is an easy counterexample of safety of the proposed controller design: Since the controller is only effective when red cells occur, it does not prevent the yellow cells to expand in all the cells. Although this is already a bug of the system, it may not conflict with the authors' original goal of controlling the red cells. However, a more serious bug is found by solving the following formula:
{\footnotesize 
\begin{eqnarray}\label{fourthform}\forall x (\underbrace{(\exists \vec y (\bigwedge_{i=1}^{18} y_i(y_i-a)(y_i-a^2)=0)\wedge x=F_x(\vec y))}_{\varphi_2}\rightarrow \underbrace{\neg(x=1)}_{\mbox{``not red''}})
\end{eqnarray}
}Formula (\ref{fourthform}) expresses the desirable property that when none of the neighbor cells of $x$ is red, $x$ never becomes red. However, we found again that $\llbracket \varphi_2\rrbracket = \llbracket x^4-x=0\rrbracket$, which means in this scenario the $x$ cell can still turn red. Thus, the formal model is inconsistent with the informal specification of the system, which says that non-red cells can never interact to generate red cells. In fact, the authors mentioned that the dynamics $F_x$ is not verified because of the combinatorial explosion. Finally, to give a counterexample of the design, we solve the formula
{\footnotesize 
\begin{eqnarray}\varphi_3 =_{df} \exists \vec y\exists x.(x=1\wedge \bigwedge_{i=1}^{6} y_i(y_i-a)(y_i-a^2)=0\wedge x=F_x(\vec y))
\end{eqnarray}
}The formula checks whether there exists a configuration of $y_1,...,y_6$ which are all non-red, such that $x$ becomes red. $\varphi_3$ evaluates to true. Further, we obtain $x=1,\vec y = (a,a,a,0,0,0)$ as a witness assignment for $\varphi_3$. This serves as the counterexample (see Figure 1(c)). 

This example shows how our quantifier elimination procedure provides a practical way of verifying and debugging systems over finite fields that were previously not amenable to existing formal methods and cannot be approached by exhaustive enumeration.

\section{Conclusion}
In this paper, we gave a quantifier elimination algorithm for the first-order theory over finite fields based on the Nullstellensatz over finite fields and Gr\"obner basis computation. We exploited special properties of finite fields and showed the correspondence between elimination of quantifiers, projection of varieties, and computing elimination ideals. We also generalized the Tseitin transformation from Boolean formulas to formulas over finite fields using ideal operations. The complexity of our algorithm depends on the complexity of Gr\"obner basis computation. In an application of the algorithm, we successfully found bugs in a biological controller design, where the original authors expressed that no verification methods were able to handle the system. In future work, we expect to use the algorithm to formally analyze more systems with finite field arithmetic. The scalability of the method will benefit from further optimizations on Gr\"obner basis computation over finite fields. It is also interesting to combine Gr\"obner basis methods and other efficient Boolean methods (SAT and QBF solving). See \cite{gao09} for a discussion on how the two methods are complementary to each other.

\section*{Acknowledgement}

The authors are grateful for many important comments from Jeremy Avigad, Helmut Veith, Paolo Zuliani, and the anonymous reviewers.

\newpage

\section*{Appendix: Omitted Proofs}

\subsubsection{Proof of Lemma 2.1}

This is a consequence of the Seidenberg's Lemma (Lemma 8.13 in \cite{grobnerbook}). It can also be directly proved as follows.
\begin{proof}
We need to show $\sqrt{J+ \langle x_1^q-x_1,...,x_n^q-x_n\rangle}= J+ \langle
x_1^q-x_1,...,x_n^q-x_n\rangle.$ Since by definition, any ideal is contained in its
radical, we only need to prove $$\sqrt{J+ \langle
x_1^q-x_1,...,x_n^q-x_n\rangle}\subseteq J+ \langle x_1^q-x_1,...,x_n^q-x_n\rangle.$$

Let $R$ denote $F_q[x_1,...,x_n]$. Consider an arbitrary polynomial $f$ in the ideal $\sqrt{J+ \langle
x_1^q-x_1,...,x_n^q-x_n\rangle}$. By definition, for some integer $s$, $f^s\in J+
\langle x_1^q-x_1,...,x_n^q-x_n\rangle$. Let $[f]$ and $[J]$ be the images of,
respectively, $f$ and $J$, in $R/\langle x_1^q-x_1,...,x_n^q-x_n\rangle$ under the
canonical homomorphism from $R$ to $R/\langle x_1^q-x_1,...,x_n^q-x_n\rangle.$ For brevity we write $S=\langle x_1^q-x_1,...,x_n^q-x_n\rangle$.

Now we have $[f]^s\in [J]$, and we further need $[f]\in [J]$. We prove, by induction on the structure of polynomials, that for any $[g]\in R/S$, $[g]^q=[g]$.  
\begin{itemize}
\item If $[g]=c x_1^{a_1}\cdots x_n^{a_n}+S$ ($c\in F_q, a_i\in N$), then
$$[g]^q=(cx_1^{a_1}\cdots x_n^{a_n}+S)^q=(cx_1^{a_1}\cdots
x_n^{a_n})^q+S=cx_1^{a_1}\cdots x_n^{a_n}+S=[g].$$
\item If $[g]=[h_1]+[h_2]$, by inductive hypothesis, $[h_1]^q=[h_1], [h_2]^q=[h_2]$, and,
since any element divisible by $p$ is zero in $F_q$ ($q=p^r$), then
$$[g]^q=([h_1]+[h_2])^q=\sum_{i=0}^{q}{q\choose
i}[h_1]^i[h_2]^{q-i}=[h_1]^q+[h_2]^q=[h_1]+[h_2]=[g]$$
\end{itemize}
Hence $[g]^q=[g]$ for any $[g]\in R/S$, without loss of generality
we can assume $s<q$ in $[f]^s$. Then, since $[f]^s\in [J]$, $
[f]=[f]^q=[f]^s\cdot[f]^{q-s}\in [J].$\qed
\end{proof}

\subsubsection{Proof of Lemma \ref{formula-ideal}}
\begin{proof}
Let $\vec a\in F_q^{n+m}$ be an assignment vector for $(\vec x, \vec y)$.

If $\vec a\in \llbracket \bigwedge_{i=1}^r f_i = 0 \rrbracket$, then $f_1(\vec a)=\cdots = f_r(\vec a) = 0$ and $\vec a\in V(\langle f_1,...,f_k\rangle)$. 

If $\vec a\in V(\langle f_1,...,f_r\rangle)$, then $\bigwedge_{i=1}^r f_i(\vec a) = 0$ is true and $\vec a \in \llbracket \bigwedge_{i=1}^r f_i = 0 \rrbracket $.\qed
\end{proof}

\subsubsection{Proof of Lemma \ref{just-projection}}
\begin{proof} We show set inclusion in both directions.
\begin{itemize}
 \item
For any $\vec b \in \llbracket\exists \vec x \varphi (\vec x; \vec y)\rrbracket$, by definition, there exists $\vec a\in F_q^n$ such that $(\vec a,\vec b)$ satisfies $\varphi(\vec x; \vec y)$. Therefore, $(\vec a,\vec b) \in \llbracket\varphi(\vec x; \vec y)\rrbracket$, and $\vec b\in \pi_n (\llbracket\varphi(\vec x; \vec y)\rrbracket)$.
\item 
For any $\vec b \in \pi_n (\llbracket\varphi(\vec x; \vec y)\rrbracket)$, there exists $\vec a\in F_q^n$ such that $(\vec a, \vec b)\in \llbracket\varphi(\vec x; \vec y)\rrbracket$. By definition, $\vec b\in \llbracket\exists x \varphi (\vec x; \vec y)\rrbracket$.\qed
\end{itemize}
\end{proof}

\subsubsection{Proof of Lemma 3.3}
\begin{proof}
We have $\llbracket \bigwedge_{i\in A_x} (x_i^q-x_i=0)\wedge \bigwedge_{i\in A_y} (y_i^q-y_i=0) \rrbracket = \llbracket \top \rrbracket$, which follows from Proposition 2.1. \qed
\end{proof}

\subsubsection{Proof of Lemma \ref{no-negation}}
\begin{proof}
Let $\varphi[\psi_1/\psi_2]$ denote substitution of $\psi_1$ in $\varphi$ by $\psi_2$. Suppose the negative atomic formulas in $\varphi$ are $f_1\neq 0,...,f_k\neq 0$. 

We introduce a new variable $z_1$, and substitute $f_1\neq 0$ by $p\cdot z_1 = 1$. Since the field $F_q$ does not have zero divisors, all the solutions for $\llbracket f_1\neq 0\rrbracket = \llbracket\exists z_1 (p\cdot z_1 =1)\rrbracket$ (the Rabinowitsch trick). 

Iterating the procedure, we can use $k$ new variables $z_1,...,z_k$ so that:
$$\llbracket\varphi\rrbracket = \llbracket \varphi[f_1\neq 0 /(\exists z_1. (p\cdot z_1 -1 = 0))]\cdots [f_k\neq 0 /(\exists z_k. (p\cdot z_k -1 = 0))]\rrbracket$$
 Since the result formula contains no more negations and the $z_i$s are new variables, it can be put into prenex form $\exists \vec z. (\varphi[f_1\neq 0 /(p\cdot z_1 -1 = 0)]\cdots [f_k\neq 0 /(p\cdot z_k -1 = 0)])$.\qed
\end{proof}

\subsubsection{Proof of Lemma \ref{no-disjunction}}
\begin{proof}$\llbracket \psi_1\vee \psi_2 \rrbracket = V(J_1)\cup V(J_2)$ follows from the definition of realization. We only need to show the second equality. Let $\vec a=(a_1,...,a_n)\in F_q^n$ be a point.

- Suppose $\vec a\in V(J_1)\cup V(J_2)$. If $\vec a\in V(J_1)$, then $(1, a_1,...,a_n)\in V(x_0J_1+(1-x_0)J_2)$. If $\vec a\in V(J_2)$, then $\langle 0, a_1,...,a_n\rangle\in V(x_0J_1+(1-x_0)J_2)$. In both cases, $\vec a\in \pi_0(V(x_0J_1+(1-x_0)J_2))$.

- Suppose $\vec a \in \pi_0( V(x_0J_1+(1-x_0)J_2))$. There exists $a_0\in F_q$ such that $( a_0,a_1,...,a_n)\in V(x_0J_1+(1-x_0)J_2)$. If $a_0\not\in \{0,1\}$, then all the polynomials in $J_1$ and $J_2$ need to vanish on $\vec a$; if $a_0=1$ then $J_1$ vanishes on $\vec a$; if  $a_0=0$ then $J_2$ vanishes on $\vec a$. In all cases, $\vec a\in V(J_1)\cup V(J_2)$.\qed
\end{proof}

\subsubsection{Proof of Theorem 5.1}
\begin{proof}
We only need to show the intermediate formulas obtained in Step 1-3 are always equivalent to the original formula $\varphi$. In Step 1, the formula is flattened with ideal operations, which preserve the realization of the formula as proved in Theorem \ref{flatten-theorem}. In Step 2, we have (by Theorem \ref{main-theorem}) $\llbracket \exists \vec x_m\exists \vec t\exists \vec s (\bigwedge_{i=1}^r (f_i = 0))\rrbracket = \llbracket \bigwedge_{i=1}^u (g_i = 0) \rrbracket.$ 

Hence the formula obtained in Step 2 is equivalent to $\varphi$. In Step 3, the substitution preserves realization of the formula because
$$\llbracket \bigwedge_{i=1}^{u}  \forall \vec x_{m-1} (g_i=0)  \rrbracket = 
\llbracket \bigwedge_{i=1}^u (\neg \exists  \vec x_{m-1} (\neg g_i=0)) \rrbracket  = 
\llbracket (\bigwedge_{i=1}^{u} (\bigvee_{j=1}^{v_i} h_{ij}\neq 0))\rrbracket,$$
where the second equality is guaranteed by Theorem \ref{main-theorem} again.

The loop terminates either at the end of Step 2 or Step 3. Hence the output quantifier-free formula $\psi$ is always in conjunctive normal form, which contains only variables $\vec y$, and is equivalent to the original formula $\varphi$.\qed
\end{proof}
\end{document}